\documentclass[10pt,a4paper]{article}

\usepackage[utf8]{inputenc}
\usepackage{hyperref}
\usepackage{subcaption}
\usepackage{float}
\usepackage{graphicx} 
\usepackage{amsmath, amsthm, amssymb} 
\DeclareMathOperator*{\argmin}{arg\,min} 
\usepackage{multirow}
\usepackage{float}
\usepackage[dvipsnames]{xcolor}

\newtheorem{Definition}{Definition}
\newtheorem{Theorem}{Theorem}

\usepackage{multicol}
\newtheorem{Lemma}{Lemma}

\title{Quantum Annealing for Computer Vision Minimization Problems}
\author{Shahrokh  Heidari~* \and Michael J.~Dinneen~* \and Patrice Delmas \thanks{School of Computer Science,
The University of Auckland, New Zealand}}

\begin{document}

\maketitle

\begin{abstract}
    \label{abstract}
Computer Vision (CV) labelling algorithms play a pivotal role in the domain of low-level vision. For decades, it has been known that these problems can be elegantly formulated as discrete energy minimization problems derived from probabilistic graphical models (such as Markov Random Fields). Despite recent advances in inference algorithms (such as graph-cut and message-passing algorithms), the resulting energy minimization problems are generally viewed as intractable. The emergence of quantum computations, which offer the potential for faster solutions to certain problems than classical methods, has led to an increased interest in utilizing quantum properties to overcome intractable problems. 
This study investigates a new Quantum Annealing based inference algorithm for CV discrete energy minimization problems. Our contribution is focused on Stereo Matching as a significant CV labeling problem. As a proof of concept, we also use a hybrid quantum-classical solver provided by D-Wave System to compare our results with the best classical inference algorithms in the literature. 
\end{abstract}

\section{Introduction}
Computer Vision (CV) is a field of study focusing on how computers gain high-level perception from digital images or videos, which can help decision-making in real-world environments. While humans routinely interpret the environment, enabling computers to perceive the real world from its representation through images remains a largely unsolved problem. Many problems in CV are formulated as labeling problems. A CV labeling problem consists of a set of image features (such as pixels, edges, or image segments) on which we want to estimate quantities from a set of labels \cite{veksler1999efficient} (such as intensity in Image Restoration or disparity in Stereo Matching and Motion). Generally, CV labeling problems are modeled by a discrete minimization problem, where an objective function is defined to be optimized over a set of possible labelings. When this objective function measures the badness, the optimization problem is often called energy minimization, and the objective function is referred to as an energy function \cite{felzenszwalb2010dynamic}. Given the intrinsically tricky nature of CV minimization problems, researchers have always been looking for efficient algorithms to approximate the optimal solution as fast and accurately as possible. Thus, there has been significant development in minimization algorithms for CV problems from the classical methods in the 90s, such as Simulated Annealing \cite{vcerny1985thermodynamical}, Mean-field Annealing \cite{geiger1991parallel}, and Iterated Conditional Modes (ICM) \cite{besag1986statistical} to the recent state-of-the-art algorithms, such as graph-cut based \cite{boykov1998markov,birchfield1999multiway,boykov2001fast,komodakis2007approximate,wainwright2004tree,wainwright2005map} and message-passing based \cite{sun2003stereo,felzenszwalb2006efficient,yu2007efficient} approaches (we refer interested readers to the most recent comparative studies on CV minimization algorithms \cite{tappen2003comparison,szeliski2006comparative,kolmogorov2006comparison,szeliski2008comparative,kappes2015comparative}). Despite being extensively researched and even considering the most recent advances using deep learning-based strategies~\cite{voulodimos2018deep}, which are computationally expensive, CV labeling problems are still considered open problems with no prefect (optimal) solutions.  \\
 
 Therefore, researchers have always been looking for alternatives to tackle the problem.  With the advent of quantum computations which promise potentially lower-time complexity on certain problems than the best-classical counterparts~\cite{denchev2016computational,king2019quantum,yaacoby2021comparison}, recent studies have focused on leveraging quantum properties to overcome intractable classical problems using Quantum Annealing (QA). D-Wave System was the first company to build a Quantum Processing Unit (QPU) that naturally approximates the ground state of a particular problem representation, namely Ising model~\cite{mcgeoch2014adiabatic}. The importance of Ising models is that one can solve a variety of NP-hard optimization problems by finding the corresponding ground state~\cite{lucas2014ising,calude2017solving,calude2017qubo}. Despite the promising experiments~\cite{denchev2016computational,king2019quantum}, D-Wave QPUs are specifically designed to solve optimization problems, making them less versatile than other quantum computing approaches. This restricts their application domain primarily to optimization and sampling tasks, while they may not be suitable for more general-purpose computing requirements. Also, D-Wave QPUs exhibit limited qubit connectivity, and the scarcity of available qubits has been consistently challenging, from the 128-qubit D-Wave One built in 2011 to the newly released 5000-qubit D-Wave Advantage. Therefore, large CV problems involving highly non-convex functions in a search space of many thousands of dimensions have not been widely explored to see if QA can provide advantages in real-world CV problems. In recent years, there has been a growing interest in Quantum Computer Vision (QCV), largely fueled by recent advancements in D-Wave QPU architectures and their capabilities in solving optimization problems, such as Classification~\cite{adachi2015, dixit2020,koshka2016,koshka2017,koshka2018}, Synchronization \cite{birdal2021quantum,arrigoni2022quantum}, Tracking \cite{zaech2022adiabatic}, Fitting \cite{doan2022hybrid,farina2023quantum}, Detection \cite{li2020quantum}, and Matching \cite{cruz2018qubo,heidari2021improved,benkner2020adiabatic,yurtsever2022q,benkner2022quant} problems. However, each method employs a distinct quantum model to represent the respective CV problem, allowing it to be minimized on a D-Wave QPU. A versatile framework for converting a CV problem into an appropriate quantum model holds significant value. Such a flexible solution not only simplifies the process of adapting various CV problems for quantum computation but also opens up new avenues for harnessing the power of quantum computing in addressing intricate optimization tasks. 
 
 In this study, we aim to focus on a challenging labeling problem, Stereo Matching, and provide a general-purpose quantum model that can be used for any CV labeling problem (such as Image Segmentation, Image Restoration, Image Registration, Optical Flow, Object Detection, and Image Inpainting). Due to the scarcity of available qubits on the current D-Wave QPUs, we use a D-Wave hybrid quantum-classical solver to show the feasibility of the proposed quantum model once enough qubits are available. Our findings show that Quantum Annealing can offer promising results in CV applications compared to the state-of-the-art CV minimization inference algorithms. 
 
 The paper is organized as follows: In Section \ref{Stereo_Matching}, we provide a brief introduction to Stereo Matching, an important CV labeling problem. In Section \ref{Quantum_Annealing}, we shift the focus to Quantum Annealing and the D-Wave QPU. Our general-purpose quantum solution to Stereo Matching and its proof of correctness are presented in Section \ref{Quantum_Stereo_Matching}. We provide experimental results and numerical evaluation in Section \ref{Expriments}. Finally, Section \ref{Conclusion} concludes the paper.

\section{Stereo Matching}
\label{Stereo_Matching}
The characteristics of binocular vision in humans allow for the simultaneous observation of a singular object by both eyes. This ability significantly contributes to the understanding of depth in the brain. The distance between our eyes, often referred to as ``baseline", facilitates slight variation in the perspective captured by each eye.  Despite each eye observing a nearly identical image, a marginal displacement exists. The brain uses this displacement to perceive a 3D observation from the scene. Likewise, a stereo vision system is designed to replicate human vision mechanisms. This system comprises two horizontal cameras on the left and right sides, effectively simulating human binocular perception. Each camera in the system records an image that, while fundamentally similar, features a certain degree of displacement. This displacement, often called disparity, signifies the difference in the position of a 3D point, as observed from two different viewpoints (the left and right viewpoints)~\cite{hamzah2016literature}.  By implementing a stereo vision system, the main goal is to construct a 3D model using the left and right stereo images. This procedure may encompass various stages, including Camera Calibration (optional), Rectification, Stereo Matching, and 3D Reconstruction~\cite{malekabadi2019disparity}. Identifying corresponding 2D projections in stereo images is known as Stereo Matching. This process represents the most computationally demanding component of a stereo vision system. When the left image serves as the reference, the goal of Stereo Matching is to find the corresponding pixel in the right image for every pixel in the left image. Each matched pixel pair represent projections of a 3D point in the real world. Stereo Matching methods are broadly categorized into global and local approaches. While local methods prioritize speed, often at the cost of accuracy due to susceptibilities like local ambiguities and occlusions, global methods comprehensively consider the entire image during disparity computation. Although computationally demanding, they effectively address challenges such as occluded and textureless regions~\cite{hamzah2016literature}. These methods typically lean on probabilistic graphical models, a potent blend of probability and graph theory, for their formalism~\cite{wang2013markov}. Based on the defined probabilistic graphical model, an energy function is modeled which can be minimized to solve the Stereo Matching problem~\cite[p.~1612]{wang2013markov}. In the following, we provide the general form of a global Stereo Matching energy function, which can be adapted for any CV labeling problem (see the recent comparative study on CV labeling problems~\cite{kappes2015comparative} for more information). \\

Let $I_l$ and $I_r$ be a pair of $n\times m$ stereo images, and $D=\{d_{min}, \dots, d_{max}\}$ be a set of positive integers, where $d_{min}$ and $d_{max}$ are the lowest and highest possible disparity values, respectively. Considering the left image $I_l$ as the reference for which we want to compute a disparity map, the set of pixels is defined as (\ref{eq_P}). Here, we also initialize $N$ as a 4-neighborhood system defined in (\ref{eq_N}).

\begin{equation}
\label{eq_P}
\small
    P=\{(i,j)\mid i\in\{0,\dots, n-1\}, j\in\{0,\dots,m-1\}\},
\end{equation}
{\small
\begin{align}
\label{eq_N}
   N &= \{\{(i,j),(i',j')\} \mid (i,j)\in P,\\
   &(i',j')\in \{(i\pm 1, j), (i,j\pm1)\}\},\nonumber
\end{align}
}
\noindent%
where, $0\leq i'< n, 0\leq j' < m$. In a global Stereo Matching model, the Stereo Matching problem is modeled by a labeling problem where each pixel in $P$ is labeled by a disparity value in $D$~\cite[p.~5]{veksler1999efficient}. In fact, a ``labeling" involves mapping from $P$ to $D$. Such a labeling problem is defined by a discrete optimization problem, where an energy function is defined to be minimized over a set of possible labelings. This energy function has two terms. The first term penalizes the solutions when they are inconsistent with the data, and the second term imposes some constraints on spatial coherence~\cite[p.~1]{szeliski2008comparative}. Let $\mathbf{w}\in D^{n\times m}$ be a vector of variables defined as $\mathbf{w}=(w_{i,j})_{(i,j)\in P}$, where $w_{i,j}\in D$. The global Stereo Matching energy function $F:D^{n\times m} \rightarrow \mathbb{R}^+$ is defined as (\ref{eq_f}). 
\begin{align}
\label{eq_f}
 F(\mathbf{w})&=\sum_{(i,j)\in P}\theta_{\{i,j\}}(w_{i,j})+\lambda\sum_{\{(i,j),(i',j')\}\in N}\delta(w_{i,j}, w_{i',j'})
\end{align}
where,
\begin{align*}
&\theta_{i,j}(w_{i,j})= \left|I_l(i,j)-I_r(i-w_{i,j},j)\right|,\\\nonumber
&\delta(w_{i,j}, w_{i',j'})= 
\left\{%
\begin{array}{ll}
   0 , & \hbox{if $w_{i,j}=w_{i',j'}$;} \\
   1 , & \hbox{otherwise.} \\
\end{array}%
\right.
\end{align*}
\noindent The first term is the Sum of Absolute Difference (SAD) matching cost function defined by $\theta_{\{i,j\}}:D\rightarrow \mathbb{R}^+$. When $\theta_{\{i,j\}}(w_{i,j})$ is (or close to) zero, it means the pixel $(i,j)$ in the left image matches the pixel $(i-w_{i,j},j)$ in the right image, and they are more likely to be the same projections of a 3D point in the real world. In the second term, $\delta:D^2\rightarrow \{0,1\}$ is the penalty function that penalizes the variation of the disparities, adding one when the allocated disparities to a pair of neighboring pixels are not equal and zero otherwise. The second term assumes that the disparities of a neighborhood of pixels present some coherence and generally do not change abruptly~\cite{li1995markov}. Furthermore, $\lambda\in \mathbb{R}^+$, known as the smoothness factor, weighs the penalties given by the second term. 

We aim to provide a general-purpose quantum model for the defined global Stereo Matching problem, which can be used for any CV labeling problem. Thus, we first give the preliminaries to describe this quantum model. 

\section{Quantum Annealing}
\label{Quantum_Annealing}
Quantum Annealing ~\cite{farhi2000quantum} is a specialized optimization technique that leverages principles from quantum mechanics to solve complex computational problems. In this model, quantum bits (qubits) are particles in a quantum dynamical system that evolve based on special forces acting on them. These forces are either internal (from interactions among qubits) or external (from other sources). Each state of a register of qubits has energy based on the applied forces. A time-dependent Hamiltonian is a mathematical representation of a system, providing information about the system's energy and detailing the forces acting upon it at any given time~\cite{mcgeoch2014adiabatic}. Quantum Annealing is a computational technique employed to discover the state of the system with the minimum energy as determined by the time-dependent Hamiltonian. Consequently, Quantum Annealing constitutes a computational paradigm known for its efficiency in addressing optimization problems and providing approximations to the optimal solutions. It is inspired by the concept of annealing in metallurgy, where a material is slowly cooled to minimize defects and reach a low-energy state. In Quantum Annealing, this cooling process is simulated by a QPU known as a quantum annealer which is based on a time-dependent Hamiltonian $H(t)$ that has three components~\cite{mcgeoch2014adiabatic}: {Initial Hamiltonian} $H_I$, where all qubits are in a superposition state. \textit{Problem Hamiltonian} $H_p$, where the specific forces are defined to encode the objective function. The lowest-energy state of $H_p$ is the solution that minimizes the objective function. \textit{Adiabatic path} $s(t)$, which is a smooth function that decreases from $1$ to $0$, such as $s(t)=1-\frac{t}{t_f}$, where $s(t)$ decreases from $1$ to $0$ as $t$ increases from $0$ to some elapsed time $t_f$. During Quantum Annealing, the \textit{Initial Hamiltonian} is slowly 
evolved along the \textit{Adiabatic path} to the \textit{Problem Hamiltonian} as  $H(t)=s(t)H_I + (1-s(t))H_p$~\cite{mcgeoch2014adiabatic}, decreasing the influence of $H_I$ over time to reach $H_P$ as $s(t)$ goes from $1$ to $0$. D-Wave Systems was the first company to build a quantum annealer. To minimize/maximize an objective function using Quantum Annealing and a D-Wave QPU, it should be in a standard model like Ising or  Quadratic Unconstrained Binary Optimization (QUBO) models~\cite{mcgeoch2014adiabatic}.  Given a vector of $n$ binary variables as $\mathbf{x}=(x_1,x_2,\dots,x_{n})\in\{0,1\}^{n}$, a QUBO model is represented as $H_{qubo}(\mathbf{x})= \mathbf{x}^{T}\mathbf{Q}\mathbf{x}$, where $\{0,1\}^{n}$ is a set of $n$ binary values, and $\mathbf{Q}$ is an $n \times n$ matrix that can be chosen to be upper-diagonal. Therefore, $H_{qubo}(\mathbf{x})$ can be reformulated as (\ref{eq-Hqubo}).
\begin{equation}
\label{eq-Hqubo}
  H_{qubo}(\mathbf{x})=\sum_i\mathbf{Q}_{i,i}x_i+\sum_{i<j}\mathbf{Q}_{i,j}x_ix_j.  
\end{equation}
\noindent The diagonal terms $\mathbf{Q}_{i,i}$ are the linear coefficients acting as the external forces, and the off-diagonal terms $\mathbf{Q}_{i,j}$ are the quadratic coefficients for the internal forces~\cite{mcgeoch2014adiabatic}.

\section{Quantum Stereo Matching}
\label{Quantum_Stereo_Matching}
We introduce an equivalent QUBO model to the global Stereo Matching minimization problem (\ref{eq_f}) and provide proof of its validity. Our idea draws inspiration from the approach employed by the D-Wave Ocean SDK when handling discrete objective functions \cite{DWave_onehot}. We first allocate $|D|$ binary variables to each pixel $(i,j)\in P$, where $|D|$ is the number of elements in $D$,  $0\leq i \leq n-1$, and $0\leq j \leq m-1$. Therefore, we define $\mathbf{x}\in\{0,1\}^{nm|D|}$ as a vector of $nm|D|$ binary variables  such that $\mathbf{x}=(x_{i,j,d})$ for all  $(i,j)\in P$ and $d\in D$.  Let our QUBO model be defined as (\ref{eq_dqm_qubo}).
\begin{align}
    \label{eq_dqm_qubo}
    H(\mathbf{x})&= \alpha\sum_{(i,j)\in P}\left(1-\sum_{d\in D}x_{i,j,d}\right)^2 + \sum_{(i,j)\in P}\sum_{d\in D}\theta_{\{i,j\}}(d)x_{i,j,d} \\\nonumber
    &+ \lambda\sum_{\{(i,j),(i',j')\}\in N}\sum_{d_1\in D}\sum_{d_2\in D}\delta(d_1,d_2)x_{i,j,d_1}x_{i',j',d_2},
\end{align}
\noindent where $\alpha>\left(\sum_{(i,j)\in P}\max\{\theta_{\{i,j\}}(d)\mid d\in D\}\right)+\lambda|N|$, and $|N|$ is the number of elements in $N$. We set $\mathbf{x}^* =\argmin_{\mathbf{x}} H(\mathbf{x})$ and define a vector of $nm$ integer values as $\mathbf{w}^* =(w^*_{i,j})_{(i,j)\in P}$, where $w^*_{i,j}=d$ if $x^*_{i,j,d} =1$. Then, $\mathbf{w}^*$ minimizes the global Stereo Matching energy function (\ref{eq_f}). 

\subsection*{Proof of Correctness}
Eq.~\ref{eq_dqm_qubo} has three parts. The first part guarantees each pixel is assigned a unique disparity value from $D$. The second calculates the cost of the assigned disparity values to the pixels. The third part encodes the defined contextual constraint. 
\begin{Definition}
\label{def_dqm_feasible}
$\mathbf{x}$ is called feasible if and only if  $\sum_{d\in D} x_{i,j,d} = 1$ for all pixels $(i,j)\in P$. We denote a feasible $\mathbf{x}$ by $\mathbf{x}'$. 
\end{Definition}
\noindent Definition~\ref{def_dqm_feasible} states that given a pixel $(i,j)\in P$, its corresponding vector of binary variables $(x'_{i,j,d_{min}}, \dots, x'_{i,j,d_{max}})$ has only one value of ``1" in its values, making it possible to label each pixel uniquely by a disparity $d\in D$. Hence, the allocated disparity to a pixel $(i,j)\in P$ is $d$ if $x'_{i,j,d}=1$.
\begin{Definition}
\label{def_dqm_labeling}
Given $\mathbf{x}'$, the corresponding integer vector $\mathbf{w}'=(w'_{i,j})_{(i,j)\in P}$ is called a \emph{labeling}, where $w'_{i,j}=d$ if $x'_{i,j,d}=1$.
\end{Definition}
\begin{Lemma}
    \label{lemma_dqm_f}
    Given a feasible $\mathbf{x}'$ and its corresponding labeling $\mathbf{w}'$, the equality $H(\mathbf{x}') = F(\mathbf{w}')$ holds, where $F$ is the global Stereo Matching energy function in (\ref{eq_f}).
\end{Lemma}
\begin{proof} Considering $H(\mathbf{x}')$ in (\ref{eq_dqm_qubo}), 
    \begin{itemize}
        \item Since $\mathbf{x}'$ is feasible, $\sum_{d\in D} x'_{i,j,d} = 1$ for all pixels $(i,j)\in P$ by Definition~\ref{def_dqm_feasible}. Therefore, we have
        \[
        \alpha\sum_{(i,j)\in P}\left(1-\sum_{d\in D}x'_{i,j,d}\right)^2 = \alpha\sum_{(i,j)\in P}\left(1-1\right)^2 = 0. 
        \]
        \item Given a pixel $(i,j)\in P$, only one variable in the vector $(x'_{i,j,d_{min}},\dots, x'_{i,j,d_{max}})$ is one, and all the others are zero. This non-zero variable is $x'_{i,j,w'_{i,j}}$ by Definition~\ref{def_dqm_labeling}. Therefore, we have
        \begin{align*}
        \sum_{(i,j)\in P}\sum_{d\in D}\theta_{\{i,j\}}(d)x'_{i,j,d} &= \sum_{(i,j)\in P}\theta_{\{i,j\}}(w'_{i,j})x'_{i,j,w'_{i,j}}= \sum_{(i,j)\in P}\theta_{\{i,j\}}(w'_{i,j}).
        \end{align*}
        \item Given $\{(i,j),(i',j')\}\in N$, the two corresponding vectors of binary variables are 
        \begin{itemize}
            \item $(i,j): (x_{i,j,d_{min}},\dots, x_{i,j,d_{max}})$
            \item $(i',j'): (x'_{i',j',d_{min}},\dots, x'_{i',j',d_{max}})$
        \end{itemize}
         Since $\mathbf{x}'$ is feasible, only one of the variables in each vector is one, and the others are zero. These variables are $x'_{i,j,w'_{i,j}}$ and $x'_{i',j',w'_{i'j'}}$, respectively, by Definition~\ref{def_dqm_labeling}. Thus, we can write

        \begin{align*}
        \tiny
        &\lambda\sum_{\{(i,j),(i',j')\}\in N}\sum_{d_1\in D}\sum_{d_2\in D}\delta(d_1,d_2)x'_{i,j,d_1}x'_{i',j',d_2} \\
        &= \lambda\sum_{\{(i,j),(i',j')\}\in N}\delta(w'_{i,j},w'_{i'j'})x'_{i,j,w'_{i,j}}x'_{i',j',w'_{i'j'}}\\
        &= \lambda\sum_{\{(i,j),(i',j')\}\in N}\delta(w'_{i,j},w'_{i'j'}). 
        \end{align*}
    \end{itemize}
    Therefore, we can rewrite $H(\mathbf{x}')$ as follows. 
    \begin{align*}
    H(\mathbf{x}') &= \sum_{(i,j)\in P}\theta_{\{i,j\}}(w'_{i,j}) + \lambda\sum_{\{(i,j),(i',j')\}\in N}\delta(w'_{i,j},w'_{i'j'}) = F(\mathbf{w}').
    \end{align*}
\end{proof}
\begin{Lemma}
\label{lemma_xsfeasible}
Let $\mathbf{x}^* = \argmin_{\mathbf{x}}H(\mathbf{x})$. $\mathbf{x}^*$ is feasible. 
\end{Lemma}
\begin{proof}
    For ease of reference, we rewrite $H(\mathbf{x})$ as follows:
    \[
    H(\mathbf{x}) = \alpha \mathcal{A}(\mathbf{x}) + \mathcal{B}(\mathbf{x}),
    \]
    \noindent where
    \begin{align*}
        \mathcal{A}(\mathbf{x}) &= \sum_{(i,j)\in P}\left(1-\sum_{d\in D}x_{i,j,d}\right)^2\\
        \mathcal{B}(\mathbf{x}) &=  \sum_{(i,j)\in P}\sum_{d\in D}\theta_{\{i,j\}}(d)x_{i,j,d} + \lambda\sum_{\{(i,j),(i',j')\}\in N}\sum_{d_1\in D}\sum_{d_2\in D}\delta(d_1,d_2)x_{i,j,d_1}x_{i',j',d_2}.
    \end{align*}
    Towards a contradiction, suppose that $\mathbf{x}^*$ is not feasible. In this case, $\mathcal{A}(\mathbf{x}^*)\neq 0$ and it is non-negative. Therefore,  
    \begin{equation}
    \label{eq_proof_1}
     H(\mathbf{x}^*) = \alpha\mathcal{A}(\mathbf{x}^*) + \mathcal{B}(\mathbf{x}^*).
    \end{equation}
    \noindent Given a feasible $\mathbf{x}'$, $\mathcal{A}(\mathbf{x}') = 0$, and we have
    \begin{equation}
    \label{eq_proof_2}
        H(\mathbf{x}') = \mathcal{B}(\mathbf{x}').
    \end{equation}
     
    \noindent Since $\mathbf{x}'$ is feasible, $\mathcal{B}(\mathbf{x}')$ 
    adds penalty values up to a maximum of $$\left(\sum_{(i,j)\in P}\max\{\theta_{i,j}(d) \mid d\in D\}\right)+\lambda|N|.$$ Considering (\ref{eq_proof_2}), we have
    \begin{equation}
    \label{eq_proof_3}
        H(\mathbf{x}') \leq \left(\sum_{(i,j)\in P}\max\{\theta_{i,j}(d) \mid d\in D\}\right)+\lambda|N|.
    \end{equation}
    
    \noindent We know that $\alpha > \left(\sum_{(i,j)\in P}\max\{\theta_{i,j}(d) \mid d\in D\}\right)+\lambda|N|$, $\mathcal{A}(\mathbf{x}^*)$ is non-zero and non-negative, and $\mathcal{B}(\mathbf{x}^*)$ is non-negative. Considering (\ref{eq_proof_1}), we can write
    \begin{align}
        \label{eq_proof_4}
         H(\mathbf{x}^*)&= \alpha \mathcal{A}(\mathbf{x}^*) + \mathcal{B}(\mathbf{x}^*)> \left(\sum_{(i,j)\in P}\max\{\theta_{i,j}(d) \mid d\in D\}\right)+\lambda|N|.
    \end{align}
    \noindent The following statement is true by (\ref{eq_proof_3}) and (\ref{eq_proof_4}): $H(\mathbf{x}')< H(\mathbf{x}*)$, which is a contradiction. Therefore, $\mathbf{x}^*$ is feasible. 
\end{proof}
\begin{Theorem}
\label{theroem_dqm_qubo}
  Given $\mathbf{w}^*$ as the corresponding labeling of $\mathbf{x}^*$, $\mathbf{w}^*$ minimizes the global Stereo Matching energy function $F$ defined in (\ref{eq_f}).    
\end{Theorem}
\begin{proof}
    Towards a contradiction, we suppose that $\mathbf{w}^*$ does not minimize $F$. In this case, there must be a feasible $\mathbf{x'}$ for which its corresponding labeling $\mathbf{w}'$ minimizes $F$. Therefore, we have $F(\mathbf{w}') < F(\mathbf{w}^*)$. Since $\mathbf{x}'$ and $\mathbf{x}^*$ are both feasible (see Lemma~\ref{lemma_xsfeasible}), we have $H(\mathbf{x}')<H(\mathbf{x}^*)$ by Lemma~\ref{lemma_dqm_f}. This is a contradiction because in this case $\mathbf{x}^*\neq \argmin_{\mathbf{x}}H(\mathbf{x})$. 
\end{proof}

 Eq.~\ref{eq_dqm_qubo} is versatile and can be adapted for a variety of CV labeling problems by replacing $P$ with any desired set of image features, and replacing $D$ with an appropriate set of labels depending on the application. Then, the first and second terms in (\ref{eq_f}) can be defined accordingly. The modified QUBO remains consistent with the QUBO model described in  (\ref{eq_dqm_qubo}). This adaptability showcases the broader applicability of the model, making it a flexible tool for addressing a range of CV labeling challenges.

\section{Evaluation and Experimental Results}
\label{Expriments}
\subsection{Qubit Complexity}
D-Wave quantum computers have showcased remarkable potential in solving optimization problems. However, one significant challenge they face is the limited availability of qubits. D-Wave QPUs employ Quantum Annealing to find (estimate) the global minimum of a QUBO/Ising model. While effective for specific problem types, this approach often requires a large number of qubits, and the current generation of D-Wave QPUs have constraints on the number of qubits that can be utilized. Consequently, proposing a QUBO model with fewer variables is paramount as it addresses the current limitations in qubit availability, enables the solution of larger and more complex problems, widens access to Quantum Annealing, and enhances the robustness and practicality of Quantum Annealing technology in solving real-world optimization challenges. Recall $P$ as the set of pixels for a pair of stereo images with size $n\times m$, and $D$ as the set of possible disparities values, where $|P|=nm$ and $|D|=k$ denotes the number of elements in $P$ and $D$, respectively. Given the defined vector of binary variables in (\ref{eq_dqm_qubo}), the number of QUBO variables in our general-purpose quantum model is $nmk$. Table~\ref{tab_comparison_qubit_complexity} compares our quantum model with the existing labeling-based quantum solutions that can be utilized for Stereo Matching. Cruz-Santos et al.~\cite{cruz2018qubo} and Heidari et al.\cite{heidari2021improved} models are based on the minimum \textit{st} cut problem, and Heidari et al.\cite{heidari2022equivalent} approach reduces a CV labeling problem to the minimum multi-way cut problem. 
\begin{table}[ht]
\centering
\caption{Qubit-complexity comparison of the proposed quantum Stereo Matching models.}
\label{tab_comparison_qubit_complexity}
\begin{tabular}{|c|c|}
\hline
\textbf{Model} & \textbf{Qubit complexity} \\
\hline
\cite{cruz2018qubo} & $7nmk+9nm-2nk-2mk-2n-2m+2$\\
\hline
\cite{heidari2021improved} & $nmk+nm+2$\\
\hline
\cite{heidari2022equivalent} & $nmk+k^2$ \\
\hline
Ours& $nmk$ \\
\hline
\end{tabular}
\end{table}
\subsection{Experimental Results}
Once a QUBO model is prepared, it needs to be embedded within the QPU hardware architecture for the minimization process. Embedding is the crucial step of mapping QUBO variables onto the available qubits on the hardware. Embedding can be challenging due to the relatively limited qubits and the restricted hardware connectivity. Consequently, it is common to chain two or more qubits together on the QPU to represent a single QUBO variable. While many real-world applications can successfully run on the D-Wave QPUs, there are cases where the input data is too large to be directly solved by Quantum Annealing, primarily because of the qubit scarcity. To overcome this size limitation, hybrid solvers combine classical and quantum approaches for problem-solving. D-Wave hybrid solvers can handle problems with a significantly higher number of variables than those directly solvable by a D-Wave QPU, offering a reliable estimate of the future accuracy of D-Wave QPUs once more qubits become available on the hardware. As a proof of concept, we utilize the Constrained Quadratic Model (CQM) D-Wave hybrid solver to minimize the proposed quantum Stereo Matching model. This solver has the capability of handling up to 500,000 QUBO variables, but it still poses restrictions on the size of the input stereo images and the number of disparities that can be processed. Therefore, we had to use cropped pairs of stereo images to analyze the performance of the Stereo Matching quantum model. We chose four pairs of stereo images from 2001-Middlebury image datasets~\cite{middlebury}, namely \textit{Venus} and \textit{Bull}, \textit{Sawtooth}, and \textit{Barn}. We could not use the latest stereo datasets because of their high disparity range. Given a pair of cropped regions from recent Middlebury stereo datasets, the majority of regions in both cropped stereo images would be occluded due to a large disparity range, resulting in only a small portion of the scene being visible in both images. This makes them not suitable to evaluate our quantum model due to the simplicity of the defined global Stereo Matching energy function in (\ref{eq_f}). To identify a more common region of interest in both stereo images, we selected our pairs of stereo images from the 2001 Middlebury image dataset~\cite{middlebury}, as well as two ``natural" images (\textit{Tree} and \textit{Castle}) from the real world with low disparity ranges to incorporate complex scene structures into our prepared dataset. Figure~\ref{fig_our_dataset} illustrates our prepared stereo dataset with the corresponding ``ground truths." We did not have ground truths for the natural images since they were not created in controlled laboratory settings like 2001-Middlebury image datasets~\cite{middlebury}. Therefore, we incorporated a deep-learning-based model~\cite{li2022practical} to get fairly accurate disparity maps to be used as the corresponding ground truths. Note that we used the gray-scale versions of the shown stereo images. \\

We also selected the best-performing and state-of-the-art classical minimization algorithms commonly used in CV, so that we can compare our quantum model with the classical counterparts. The selected classical algorithms includes Swap move~\cite{boykov2001fast}, Expansion move~\cite{boykov2001fast}, Max product Loopy Belief Propagation (LBP)~\cite{felzenszwalb2006efficient},  and the improved Tree Re-weighted Message Passing (TRW-S)~\cite{kolmogorov2005convergent}. We utilized the Middlebury software framework for our classical implementations~\cite{middlebury_MRF}.
\begin{figure}
\centering
\begin{subfigure}{0.7\textwidth}
         \centering
         \includegraphics[width=\textwidth]{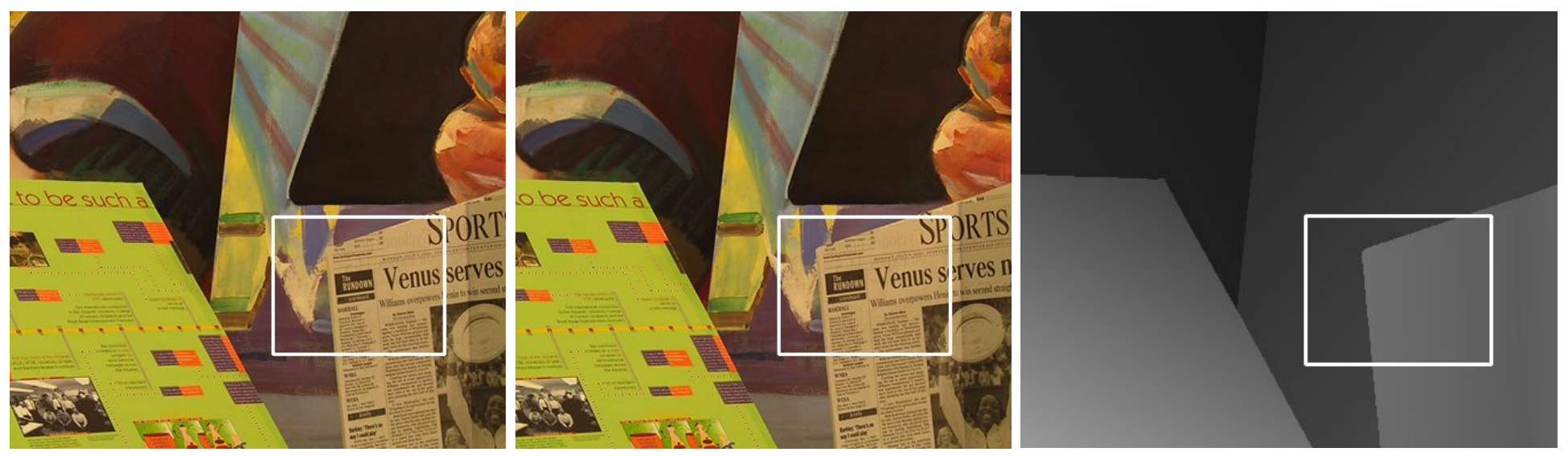}
         \caption{}
         \label{fig_our_dataset_venus}
     \end{subfigure}

\vfill
\begin{subfigure}{0.7\textwidth}
         \centering
         \includegraphics[width=\textwidth]{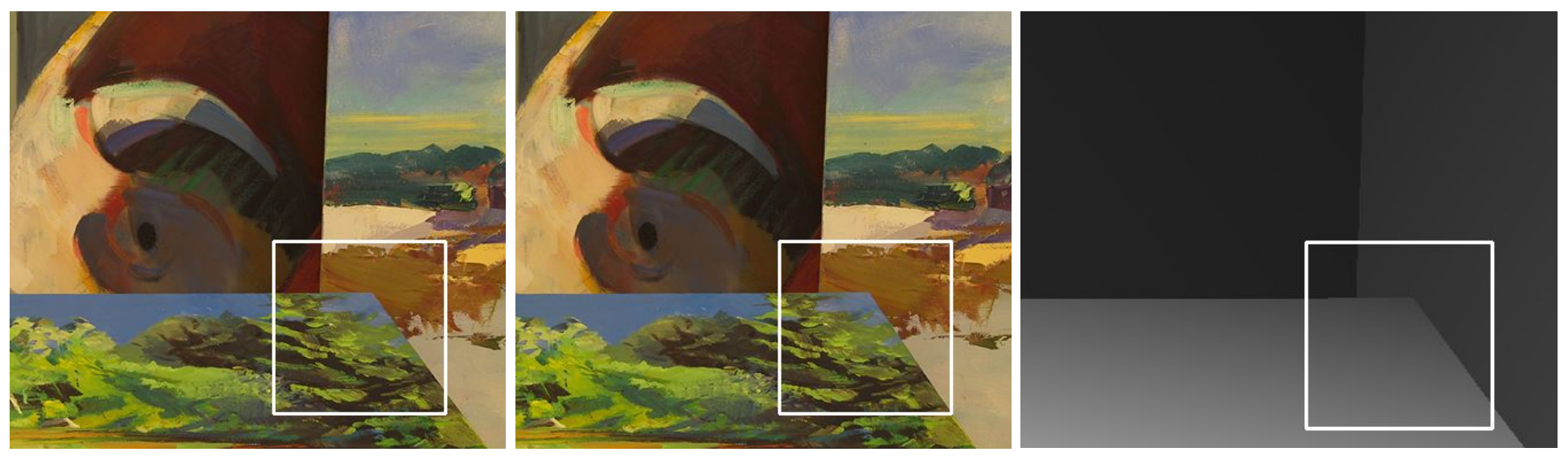}
         \caption{}
         \label{fig_our_dataset_bull}
\end{subfigure}
\vfill
\begin{subfigure}{0.7\textwidth}
         \centering
         \includegraphics[width=\textwidth]{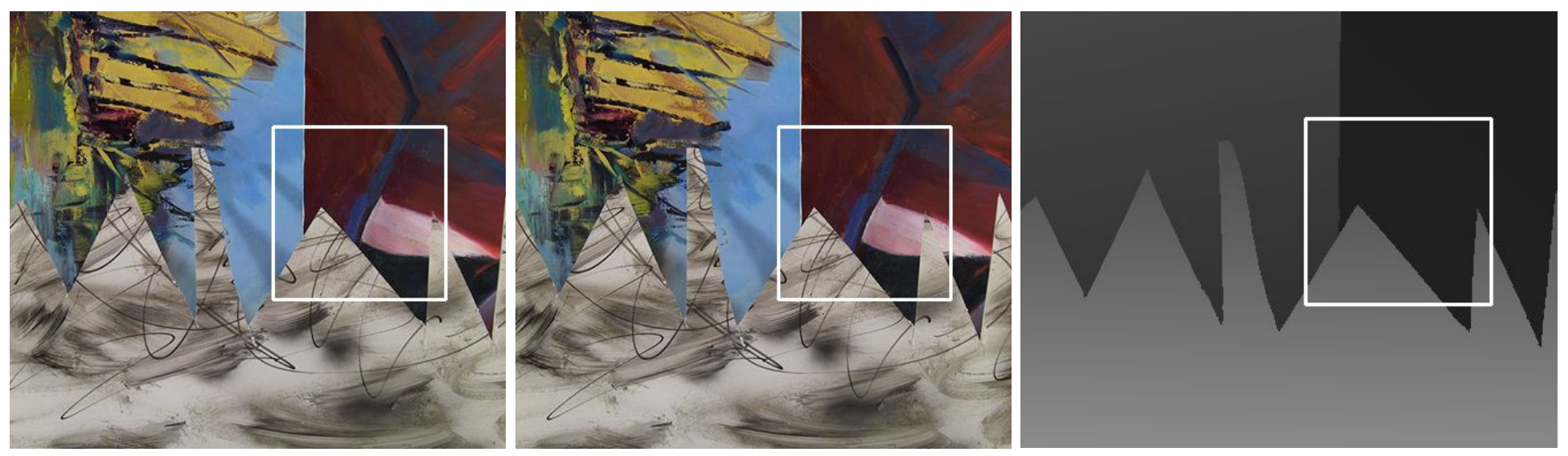}
         \caption{}
         \label{fig_our_dataset_sawtooth}
\end{subfigure}
\vfill
\begin{subfigure}{0.7\textwidth}
         \centering
         \includegraphics[width=\textwidth]{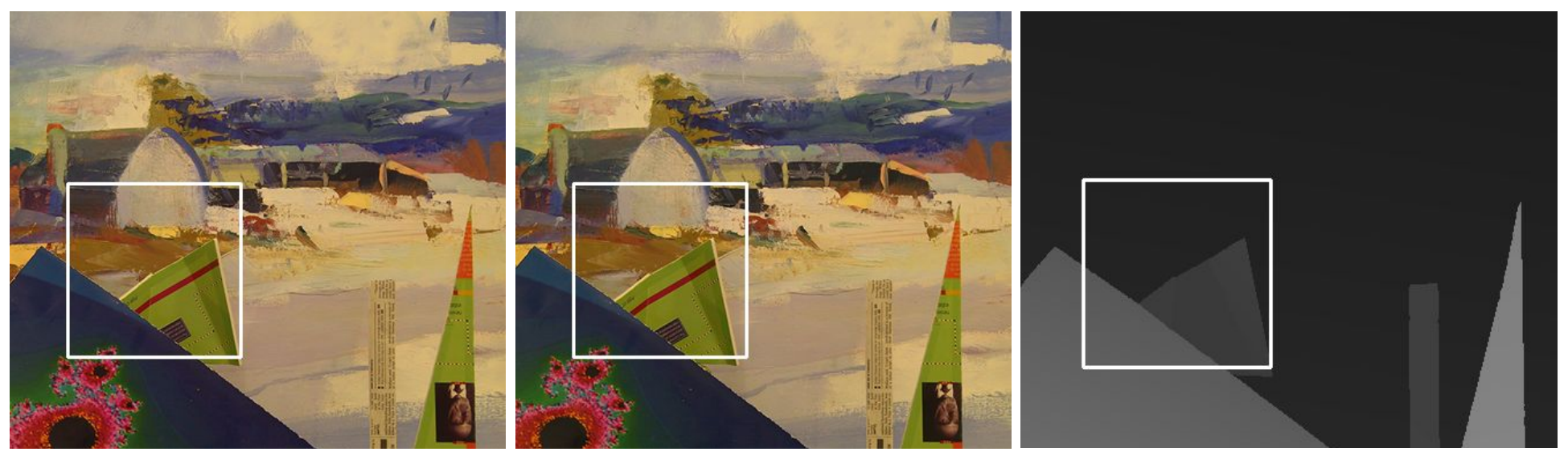}
         \caption{}
         \label{fig_our_dataset_barn}
\end{subfigure}
\vfill
\begin{subfigure}{0.7\textwidth}
         \centering
         \includegraphics[width=\textwidth]{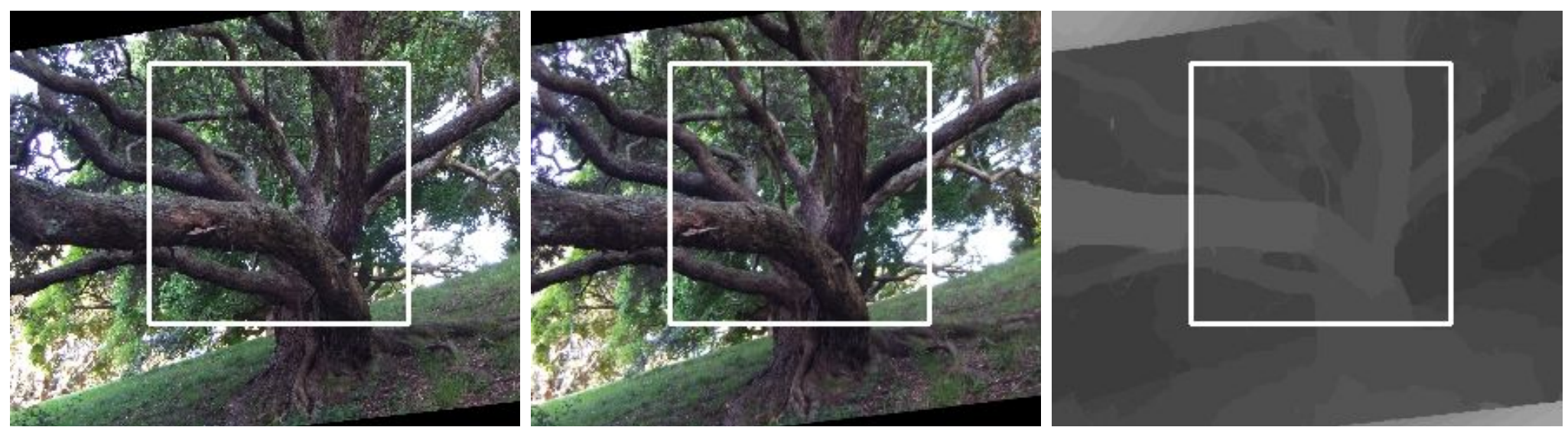}
         \caption{}
         \label{fig_our_dataset_tree}
\end{subfigure}
\vfill
\begin{subfigure}{0.7\textwidth}
         \centering
         \includegraphics[width=\textwidth]{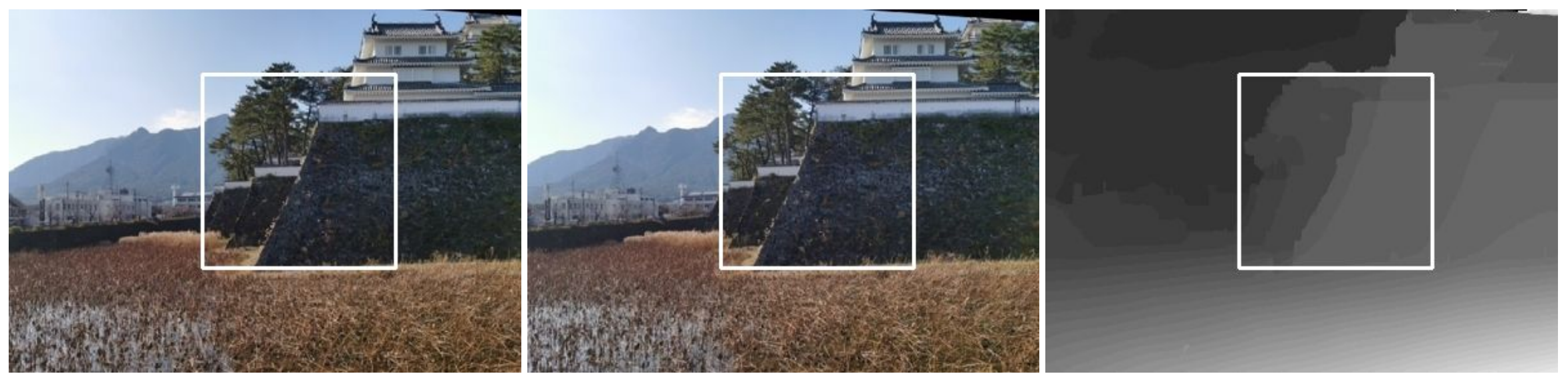}
         \caption{}
         \label{fig_our_dataset_castle}
\end{subfigure}
\caption{The prepared stereo datasets, (a) \textit{Venus}, (b) \textit{Bull}, (c) \textit{Sawtooth}, (d) \textit{Barn}, (c) \textit{Tree}, and (d) \textit{Castle}. In each row of images, we have the left stereo image, the right stereo image, and the corresponding ground truth for the left stereo image. The white squares show the cropped regions for our experiments. }
\label{fig_our_dataset}
\end{figure}
Considering the global Stereo Matching energy function defined in (\ref{eq_f}), we established an initial $\lambda=20$ for all the benchmark minimization algorithms to ensure a fair comparison. Figure~\ref{fig_disp_maps} shows the computed disparity maps by the benchmark minimization algorithms.

\begin{figure*}[t!]
\centering
\includegraphics[width=\textwidth]{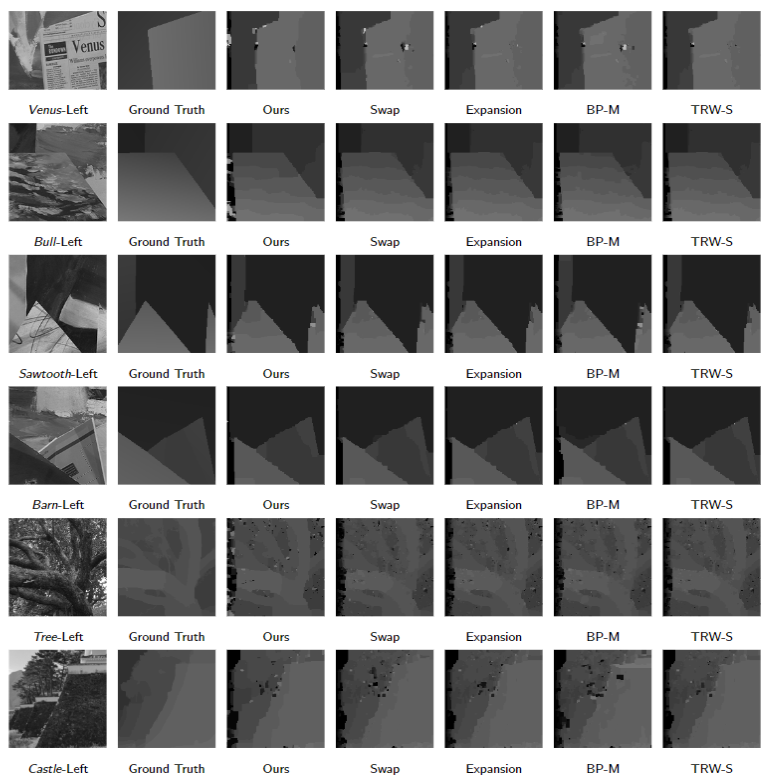}
\caption{Computed disparity maps by the benchmark minimization algorithms.}
\label{fig_disp_maps}
\end{figure*}

Next, we define two widely-used metrics~\cite{scharstein2002taxonomy} to evaluate the accuracy of our results in comparison to the corresponding ground truths: the root-mean-squared error ($rms$) and the percentage of mismatched pixels (\textit{bad-$\beta$}). Given a disparity map $\mathcal{D}$ and a ground truth $\mathcal{T}$ defined by $n\times m$ matrices of integers, $rms$ and \textit{bad-$\beta$} are defined as (\ref{eq_rms_b}). \\

\begin{equation}
\label{eq_rms_b}
    \hbox{\textit{rms}}=\sqrt{\frac{1}{nm}\sum_{i=1}^{n}\sum_{j=1}^{m} \left(\mathcal{D}_{i,j}-\mathcal{T}_{i,j}\right)^2},
\end{equation}
\begin{equation}
    \hbox{\textit{bad-$\beta$}}=\left(\frac{1}{nm}\sum_{i=1}^{n}\sum_{j=1}^{m} \left(|\mathcal{D}_{i,j}-\mathcal{T}_{i,j}| > \beta \right)\right) \times 100
\end{equation}
where $\beta\in \mathbb{R}^+$ is the disparity error tolerance. In the following evaluation, we set $\beta$ to $0.5$ and $1.0$, named \textit{bad-$0.5$} an \textit{bad-$1.0$}, respectively. Table.~\ref{tab_evaluation} compares the performance of each minimization algorithm based on the defined metrics. Given the cropped stereo images, the results suggest that the quantum model outperformed the classical counterparts on the \textit{Bull}, \textit{Sawtooth}. and \textit{Tree} datasets, and performed competitively on the \textit{Venus}, \textit{Barn}, and \textit{Castle} datasets. Our findings show that Quantum Annealing can offer promising results in CV applications compared to the state-of-the-art CV minimization inference algorithms. Due to the scarcity of available qubits on the current D-Wave QPUs, we were not able to use a pure Quantum Annealing minimization, and we used a D-Wave hybrid solver, which offers a reliable estimate of the future accuracy of D-Wave QPUs once more qubits become available on the hardware.
\begin{table}
\centering
\caption[Numerical evaluation for the prepared stereo dataset]{Numerical evaluation for the \textit{\textbf{Venus}} stereo dataset.}\label{tab_evaluation}
\begin{tabular}{|l|l|c|c|c|}
\hline
Dataset & \textbf{Method} &  \textbf{\textit{rms}} &  \textbf{\textit{bad-0.5} (\%)} & \textbf{\textit{ bad-1.0} (\%)}\\
\hline
\multirow{5}{*}{\textit{Venus}} & Ours & 2.25 &    {\color{blue}40.44} &    10.45 \\
\multirow{5}{*}{} & Swap & 2.09 &    47.57 &    10.23\\
\multirow{5}{*}{} & Expansion & 1.94 &    43.81 &     9.76\\
\multirow{5}{*}{} & BP-M & 1.96 &    47.09 &     {\color{blue}9.30}\\
\multirow{5}{*}{} & TRW-S & {\color{blue} 1.92} &    44.07 &     9.54\\
\hline
\hline
\multirow{5}{*}{\textit{Bull}} & Ours & {\color{blue}2.33} &    {\color{blue}36.07} &     {\color{blue}7.08} \\
\multirow{5}{*}{} &Swap & 2.38 &    37.43 &     7.40 \\
\multirow{5}{*}{} &Expansion & 2.38 &    37.28 &     7.36 \\
\multirow{5}{*}{} &BP-M & 2.39 &    37.42 &     7.25 \\
\multirow{5}{*}{} &TRW-S & 2.38 &    37.06 &     7.32 \\
\hline
\hline
\multirow{5}{*}{\textit{Sawtooth}} & Ours & {\color{blue}2.27} &    {\color{blue}22.54} &    {\color{blue}10.26} \\
\multirow{5}{*}{} &Swap & 2.44 &    22.76 &    10.27 \\
\multirow{5}{*}{} &Expansion & 2.41 &    22.85 &    10.44 \\
\multirow{5}{*}{} &BP-M & 2.41 &    23.59 &    10.30 \\
\multirow{5}{*}{} &TRW-S & 2.36 &    22.67 &    10.36\\
\hline
\hline
\multirow{5}{*}{\textit{Barn}} & Ours & 2.27 &    {\color{blue}14.37} &     7.41\\
\multirow{5}{*}{} &Swap & 2.23 &    16.11 &     7.51 \\
\multirow{5}{*}{} &Expansion & {\color{blue}2.21} &    16.09 &     7.55 \\
\multirow{5}{*}{} &BP-M & 2.38 &    20.21 &     8.25 \\
\multirow{5}{*}{} &TRW-S & 2.23 &    15.54 &    {\color{blue}7.33} \\
\hline
\hline
\multirow{5}{*}{\textit{Tree}} &Ours & {\color{blue}2.99} &    {\color{blue}24.99} &    {\color{blue}13.22} \\
\multirow{5}{*}{}&Swap & 3.32 &    33.85 &    14.39 \\
\multirow{5}{*}{}&Expansion & 3.27 &    31.73 &    13.23 \\
\multirow{5}{*}{}&BP-M & 3.16 &    33.45 &    13.91\\
\multirow{5}{*}{}&TRW-S & 3.12 &    31.81 &    {\color{blue}13.22} \\
\hline
\hline
\multirow{5}{*}{\textit{Castle}}&Ours & {\color{blue}2.74} &    34.52 &    17.62    \\
\multirow{5}{*}{}&Swap & 2.83 &    {\color{blue}32.85} &    17.25    \\
\multirow{5}{*}{}&Expansion & 2.76 &    33.68 &    16.94 \\
\multirow{5}{*}{}&BP-M & 2.99 &    41.12 &    21.21 \\
\multirow{5}{*}{}&TRW-S & 2.66 &    33.36 &    {\color{blue}16.62} \\
\hline
\end{tabular}
\end{table}

Since our model (\ref{eq_dqm_qubo}) is a direct equivalent to the global Stereo Matching energy function (\ref{eq_f}), its energy solution can be compared with that of the iterative classical minimization algorithms. Figure~\ref{fig_energies} shows the energies of the solutions obtained by each minimization model over the provided stereo datasets. According to the findings, our approach demonstrated a capacity to obtain solutions of lower energy in comparison to the iterative classical minimization methods for each provided stereo dataset. This observation underscores the effectiveness of QUBOs when solved by D-Wave hybrid solvers. We do not provide a comparison in terms of running time, as the classical iterative minimization algorithms were significantly faster than the D-Wave hybrid solver when minimizing the corresponding QUBO models. The reason is because of the way that a D-Wave hybrid solver works. A D-Wave hybrid solver is based on the D-Wave Hybrid Solver Service (HSS). Once a QUBO is provided to the HSS, it activates one or more heuristic solvers that run in parallel, either on a CPU or a GPU platform, to identify high-quality solutions.  Each heuristic solver comprises a classical heuristic module that navigates the search space, and a quantum module is responsible for formulating quantum queries directed to the D-Wave Advantage QPU. Solutions retrieved from the QPU assist the heuristic modules in pinpointing more viable search space regions or refining the current solutions. Each heuristic solver forwards its top solution to the HSS solver. The HSS solver then determines the best solution from the collective outputs of the heuristic solvers and relays this optimal solution back to the user~\cite{DwaveDoc_hiybrid}. Therefore, the running time is not derived from a direct QPU minimization to be compared with the classical minimization methods. We used a D-Wave hybrid solver as proof of concept to show that Quantum Annealing is capable of being used in CV labeling problems once enough qubits are available on the QPU in future. 

\begin{figure}[t!]
\centering
\includegraphics[width=0.8\textwidth]{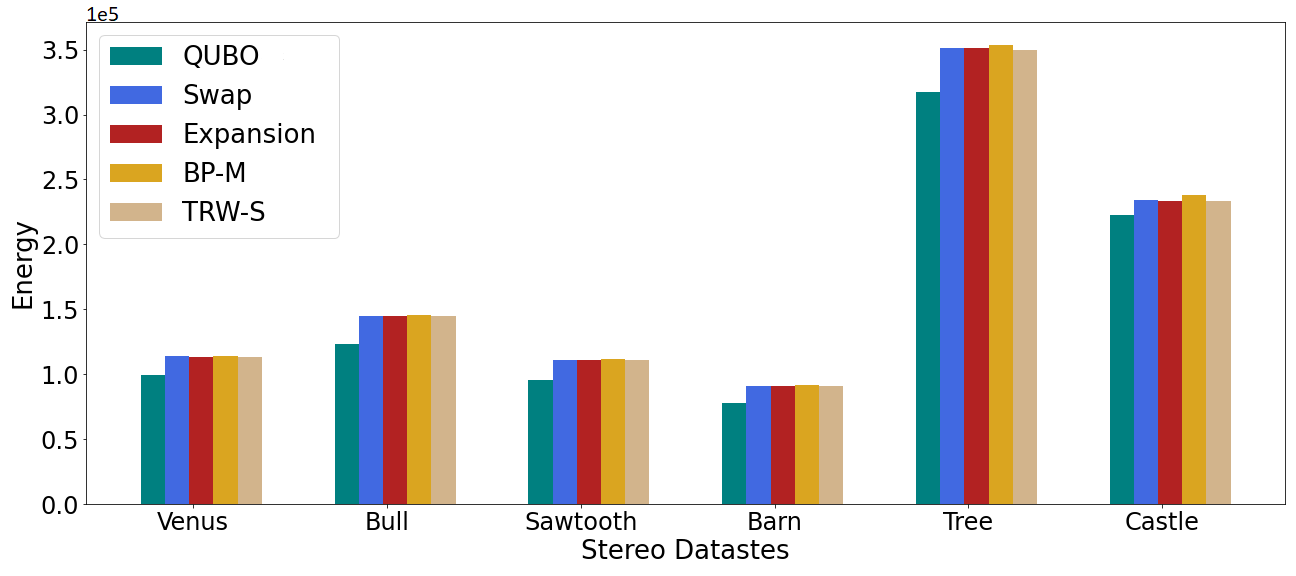}
\caption{A comparison between energies obtained by our model and the benchmarking classical minimization.}
\label{fig_energies}
\end{figure}

\section{Conclusion}
\label{Conclusion}
CV labelling algorithms play a pivotal role in the domain of low-level vision. For decades, it has been known that these problems can be elegantly formulated as discrete energy-minimization problems derived from probabilistic graphical models. Despite recent advances in inference algorithms, the resulting energy-minimization problems are generally viewed as intractable.  In this study, we presented a Quantum Annealing-based method for solving CV discrete optimization problems, specifically for Stereo Matching. However, our proposed quantum model is not limited to Stereo Matching and can be applied to various CV labeling problems such as Image Segmentation, Image Restoration, Image Registration, Optical Flow, Object Detection, and Image Inpainting. We provided proof of correctness to demonstrate the equivalence of the proposed quantum model to the original discrete minimization energy function. Due to the limited availability of qubits on the quantum hardware, we were not able to minimize the Stereo Matching energy function directly on the QPU. Instead, we utilized a D-Wave hybrid solver to show the feasibility of our proposed quantum model. Our results showed promising solutions with lower energies compared to the best classical minimization algorithms in the literature. When there are enough qubits available, it may be a subject for future research to determine if a quantum-based CV inference offers any advantages over classical minimization methods in terms of accuracy and speed.

\section*{Acknowledgements}
We thank Prof. Cristian Calude for his helpful discussions on
this paper.

\clearpage
\setcounter{page}{1}

\section{Supplementary Materials: Example*}
\label{sec_example}

This part supplements our main paper by providing a simple example to show how (\ref{eq_dqm_qubo}) can be modeled and minimized via Quantum Annealing. (\ref{fig_dqm_stereo_example_a}) and (\ref{fig_dqm_stereo_example_a}) show a pair of $(3\times4)$-sized stereo images with $D=\{0,1\}$. The intensity values for the left and right images are shown on the pixels. The corresponding pixel coordinates are illustrated in (\ref{fig_dqm_stereo_example_c}). Without loss of generality, we ignore the first column of pixels in the left image since $d_{max}$ is $1$, and we would obtain negative coordinates to match this column in the right image. The main goal is to compute the disparity map allocated to the shown red square in (\ref{fig_dqm_stereo_example_a}). Eq.~\ref{fig_dqm_stereo_example_d} shows the ground truth disparity map.
   
     \begin{figure*}[t]
        \begin{subfigure}{0.25\textwidth}
                 \centering
                 \includegraphics[width=\textwidth]{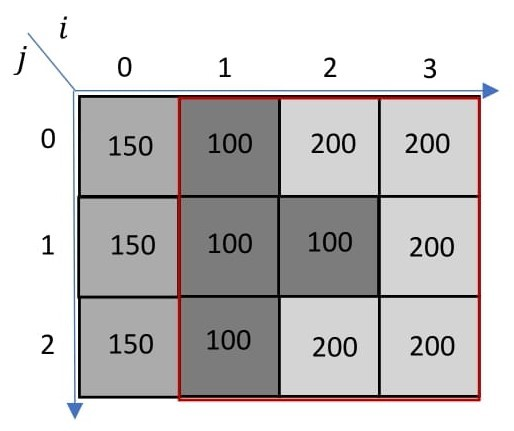}
                 \caption{}
                 \label{fig_dqm_stereo_example_a}
             \end{subfigure}
        \hspace*{0.2em}
        \begin{subfigure}{0.25\textwidth}
                 \centering
                 \includegraphics[width=\textwidth]{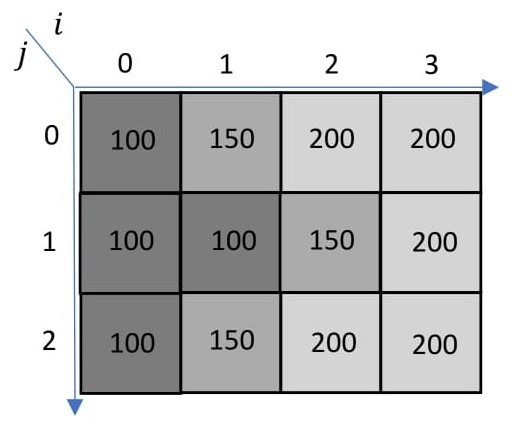}
                 \caption{}
                 \label{fig_dqm_stereo_example_b}
             \end{subfigure}
        \hspace*{0.2em}
        \begin{subfigure}{0.22\textwidth}
                 \centering
                 \includegraphics[width=\textwidth]{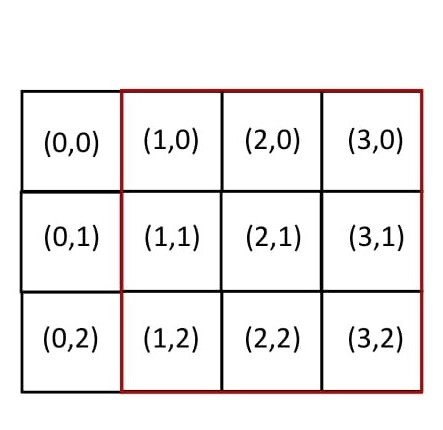}
                 \caption{}
                 \label{fig_dqm_stereo_example_c}
             \end{subfigure}  
        \hspace*{0.2em}
        \begin{subfigure}{0.17\textwidth}
                 \centering
                 \includegraphics[width=\textwidth]{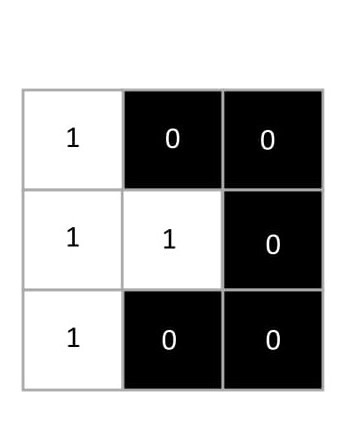}
                 \caption{}
                 \label{fig_dqm_stereo_example_d}
        \end{subfigure}     
        \caption {(a) the left stereo image, (b) the right stereo image, (c) the pixel coordinates, (d) the corresponding disparity map.}
        \label{fig_dqm_stereo_example}
    \end{figure*}
    
Considering (\ref{eq_P}) and (\ref{eq_N}), we first define $P$ and $N$ follows:
\begin{align*}
    P =& \{(1,0), (2,0), (3,0), (1,1), (2,1), (3,1), (1,2),(2,2), (3,2)\}
\end{align*}
\begin{align*}
    N =& \{ \{(1,0), (2,0)\}, \{(1,0), (1,1)\}, \{(2,0), (3,0)\},\\
    &\{(2,0), (2,1)\}, \{(3,0), (3,1)\}, \{(1,1), (2,1)\},\\
    & \{(1,1), (1,2)\}, \{(2,1), (3,1)\}, \{(2,1), (2,2)\},\\
    &\{(3,1), (3,2)\}, \{(1,2), (2,2)\}, \{(2,2), (3,2)\}\}.
\end{align*}
\noindent The numbers of pixels and disparities are $9$ and $2$, respectively. Therefore, We define a vector of $18$ binary variables as $\mathbf{x} = \{0,1\}^{18}$:
\begin{align*}
        \mathbf{x} = (&x_{1,0,0}, x_{1,0,1}, x_{2,0,0}, x_{2,0,1}, x_{3,0,0}, x_{3,0,1}, x_{1,1,0},\\
        &x_{1,1,1}, x_{2,1,0}, x_{2,1,1}, x_{3,1,0}, x_{3,1,1}, x_{1,2,0}, x_{1,2,1}, \\ &x_{2,2,0}, x_{2,2,1}, x_{3,2,0}, x_{3,2,1}).
\end{align*}
\noindent We set $\lambda=10$ and $\alpha=200$ by which we have $\alpha>\left(\sum_{(i,j)\in P}\max\{\theta_{i,j}(d)|d\in D\}\right)+\lambda|N|$. The QUBO model (\ref{eq_dqm_qubo}) is formulated as follows:
\begin{align*}
        H(\mathbf{x})&=
        200\sum_{(i,j)\in P}\left(1-\sum_{d\in D}x_{i,j,d}\right)^2 + \sum_{(i,j)\in P}\sum_{d\in D}\theta_{i,j}(d)x_{i,j,d} \\ &+10\sum_{\{(i,j),(i',j')\}\in N}\sum_{d_1\in D}\sum_{d_2\in D}\delta(d_1,d_2)x_{i,j,d_1}x_{i',j',d_2}.
\end{align*}
The QUBO model $H(\mathbf{x})$ has three terms denoted by $H_1$, $H_2$, and $H_3$ from left to right, respectively. The following shows each term's expansion separately. We then add them all at the end. We start with the first term denoted by $H_1$.
\begin{align*}
        H_1 (\mathbf{x}) = 200 (&-x_{1,0,0}-x_{1,0,1}+2x_{1,0,0}x_{1,0,1}+1\\
        &-x_{2,0,0}-x_{2,0,1}+2x_{2,0,0}x_{2,0,1}+1\\
        &-x_{3,0,0}-x_{3,0,1}+2x_{3,0,0}x_{3,0,1}+1\\
        &-x_{1,1,0}-x_{1,1,1}+2x_{1,1,0}x_{1,1,1}+1\\
        &-x_{2,1,0}-x_{2,1,1}+2x_{2,1,0}x_{2,1,1}+1\\
        &-x_{3,1,0}-x_{3,1,1}+2x_{3,1,0}x_{3,1,1}+1\\
        &-x_{1,2,0}-x_{1,2,1}+2x_{1,2,0}x_{1,2,1}+1\\
        &-x_{2,2,0}-x_{2,2,1}+2x_{2,2,0}x_{2,2,1}+1\\
        &-x_{3,2,0}-x_{3,2,1}+2x_{3,2,0}x_{3,2,1}+1)
\end{align*}
\noindent Next, we expand the second term as $H_2$.
 \begin{align*}
        H_2 (\mathbf{x})&=|I_l(1,0)-I_r(1,0)|x_{1,0,0}+|I_l(1,0)-I_r(0,0)|x_{1,0,1}\\
        &+ |I_l(2,0)-I_r(2,0)|x_{2,0,0}+|I_l(2,0)-I_r(1,0)|x_{2,0,1}\\
        &+ |I_l(3,0)-I_r(3,0)|x_{3,0,0}+|I_l(3,0)-I_r(2,0)|x_{3,0,1} \\
        &+ |I_l(1,1)-I_r(1,1)|x_{1,1,0}+|I_l(1,1)-I_r(0,1)|x_{1,1,1} \\
        &+ |I_l(2,1)-I_r(2,1)|x_{2,1,0}+|I_l(2,1)-I_r(1,1)|x_{2,1,1}\\
        &+|I_l(3,1)-I_r(3,1)|x_{3,1,0}+|I_l(3,1)-I_r(2,1)|x_{3,1,1}\\
        &+ |I_l(1,2)-I_r(1,2)|x_{1,2,0}+|I_l(1,2)-I_r(0,2)|x_{1,2,1}\\
        &+ |I_l(2,2)-I_r(2,2)|x_{2,2,0}+|I_l(2,2)-I_r(1,2)|x_{2,2,1}\\
        &+ |I_l(3,2)-I_r(3,2)|x_{3,2,0}+|I_l(3,2)-I_r(2,2)|x_{3,2,1}
\end{align*}
 \begin{align*}
        H_2 (\mathbf{x})  &= 50x_{1,0,0} + 50x_{2,0,1} +  50x_{2,1,0} + 50x_{3,1,1} +  50x_{1,2,0} +  50x_{2,2,1}.  
\end{align*}
Finally, we compute the third term as $H_3$:
\begin{align*}
        H_3 (\mathbf{x})&=10 (x_{1,0,0}x_{2,0,1}+x_{1,0,1}x_{2,0,0} \\
        &+x_{1,0,0}x_{1,1,1}+x_{1,0,1}x_{1,1,0}+x_{2,0,0}x_{3,0,1}\\
        &+x_{2,0,1}x_{3,0,0}+ x_{2,0,0}x_{2,1,1}+x_{2,0,1}x_{2,1,0}\\
        &+ x_{3,0,0}x_{3,1,1}+x_{3,0,1}x_{3,1,0}+ x_{1,1,0}x_{2,1,1}\\
        &+x_{1,1,1}x_{2,1,0}+ x_{1,1,0}x_{1,2,1}+x_{1,1,1}x_{1,2,0} \\
        &+ x_{2,1,0}x_{3,1,1}+x_{2,1,1}x_{3,1,0}+ x_{2,1,0}x_{2,2,1}\\
        &+x_{2,1,1}x_{2,2,0}+x_{3,1,0}x_{3,2,1}+x_{3,1,1}x_{3,2,0}\\
        &+x_{1,2,0}x_{2,2,1}+x_{1,2,1}x_{2,2,0}+ x_{2,2,0}x_{3,2,1}\\
        &+x_{2,2,1}x_{3,2,0}). 
\end{align*}
\noindent Adding the three terms together, we have the main QUBO model as follows:
\begin{align*}
        H(\mathbf{x}) =& -150x_{1,0,0}-200x_{1,0,1}-200x_{2,0,0}-150x_{2,0,1}\\
        &-200x_{3,0,0}-200x_{3,0,1}-200x_{1,1,0}-200x_{1,1,1}\\
        &-150x_{2,1,0}-200x_{2,1,1}-200x_{3,1,0}-150x_{3,1,1}\\
        &-150x_{1,2,0}-200x_{1,2,1}-200x_{2,2,0}-150x_{2,2,1}\\
        &-200x_{3,2,0}-200x_{3,2,1}\\
        &+400x_{1,0,0}x_{1,0,1}+400x_{2,0,0}x_{2,0,1}\\
        &+400x_{3,0,0}x_{3,0,1}+400x_{1,1,0}x_{1,1,1}\\
        &+400x_{2,1,0}x_{2,1,1}+400x_{3,1,0}x_{3,1,1}\\
        &+400x_{1,2,0}x_{1,2,1}+400x_{2,2,0}x_{2,2,1}\\
        &+400x_{3,2,0}x_{3,2,1}+1800\\
        &+10x_{1,0,0}x_{2,0,1}+10x_{1,0,1}x_{2,0,0}+10x_{1,0,0}x_{1,1,1}\\
        &+10x_{1,0,1}x_{1,1,0}+10x_{2,0,0}x_{3,0,1}+10x_{2,0,1}x_{3,0,0}\\
        &+10x_{2,0,0}x_{2,1,1}+10x_{2,0,1}x_{2,1,0}+10x_{3,0,0}x_{3,1,1}\\
        &+10x_{3,0,1}x_{3,1,0}+ 10x_{1,1,0}x_{2,1,1}+10x_{1,1,1}x_{2,1,0} \\
        &+ 10x_{1,1,0}x_{1,2,1}+10x_{1,1,1}x_{1,2,0}+ 10x_{2,1,0}x_{3,1,1}\\
        &+10x_{2,1,1}x_{3,1,0}+10x_{2,1,0}x_{2,2,1}+10x_{2,1,1}x_{2,2,0}\\
        &+10x_{3,1,0}x_{3,2,1}+10x_{3,1,1}x_{3,2,0}+10x_{1,2,0}x_{2,2,1}\\
        &+10x_{1,2,1}x_{2,2,0}+10x_{2,2,0}x_{3,2,1}+10x_{2,2,1}x_{3,2,0}. 
    \end{align*}

\noindent Giving $H(\mathbf{x})$ to the D-Wave Ocean SDK for the QPU minimization, we obtain the optimal solution $\mathbf{x}^*=\argmin_{\mathbf{x}}H(\mathbf{x})$ as follows:

\begin{multicols}{2}
    \begin{itemize}
        \item $x^*_{1,0,0}=0,$
        \item $x^*_{1,0,1}=1,$
        \item $x^*_{2,0,0}=1,$
        \item $x^*_{2,0,1}=0,$
        \item $x^*_{3,0,0}=1,$
        \item $x^*_{3,0,1}=0,$
        \item $x^*_{1,1,0}=0,$
        \item $x^*_{1,1,1}=1,$
        \item $x^*_{2,1,0}=0,$
        \item $x^*_{2,1,1}=1,$
        \item $x^*_{3,1,0}=1,$
        \item $x^*_{3,1,1}=0,$
        \item $x^*_{1,2,0}=0,$
        \item $x^*_{1,2,1}=1,$
        \item $x^*_{2,2,0}=1,$
        \item $x^*_{2,2,1}=0,$
        \item $x^*_{3,2,0}=1,$
        \item $x^*_{3,2,1}=0,$
    \end{itemize}
\end{multicols}
\noindent which match the corresponding ground-truth disparities shown in (\ref{fig_dqm_stereo_example_d}). We used the D-Wave default parameter settings for the hardware properties and initialized the number of sample-reads as $1000$. Given a pixel $(i,j)\in P$, if $x^*_{i,j,d}=1$ for $d\in D$, then $d$ is the allocated disparity to the pixel $(i,j)$. Therefore, we have the following disparities for the pixels:
\begin{multicols}{2}
        \begin{itemize}
            \item $(1,0) \leftarrow 1$
            \item $(2,0) \leftarrow 0$
            \item $(3,0) \leftarrow 0$
            \item $(1,1) \leftarrow 1$
            \item $(2,1) \leftarrow 1$
            \item $(3,1) \leftarrow 0$
            \item $(1,2) \leftarrow 1$
            \item $(2,2) \leftarrow 0$
            \item $(3,2) \leftarrow 0$
        \end{itemize}
    \end{multicols}
Eq.~\ref{fig_graphs} illustrates the corresponding D-Wave minor embedding for the defined QUBO, obtained by the D-Wave Inspector tool. Furthermore, we provide the timing information in (\ref{tab_timing}).
\begin{table}
\centering
\caption{The D-Wave timing information for the given QUBO example}
\label{tab_timing}
\begin{tabular}{|l|c|}
\hline
\textbf{Type} & \textbf{Time ($\mu s$)} \\
\hline
$qpu\_sampling\_time$ & $104760.00$\\
\hline
$qpu\_anneal\_time\_per\_sample$ & $20.00$ \\
\hline
$qpu\_readout\_time\_per\_sample$ & $64.22$ \\
\hline
$qpu\_access\_time$ & $120523.57$ \\
\hline
$qpu\_access\_overhead\_time$ & $4687.43$\\
\hline
$qpu\_programming\_time$ & $15763.57$ \\
\hline
$qpu\_delay\_time\_per\_sample$ & $20.54$\\
\hline
$total\_post\_processing\_time$ & $852.00$\\
\hline
$post\_processing\_overhead\_time$ &$ 852.00$\\
\hline
\end{tabular}
\end{table}

\begin{figure*}[t]
\centering
        \begin{subfigure}{0.385\textwidth}
                 \centering
                 \includegraphics[width=\textwidth]{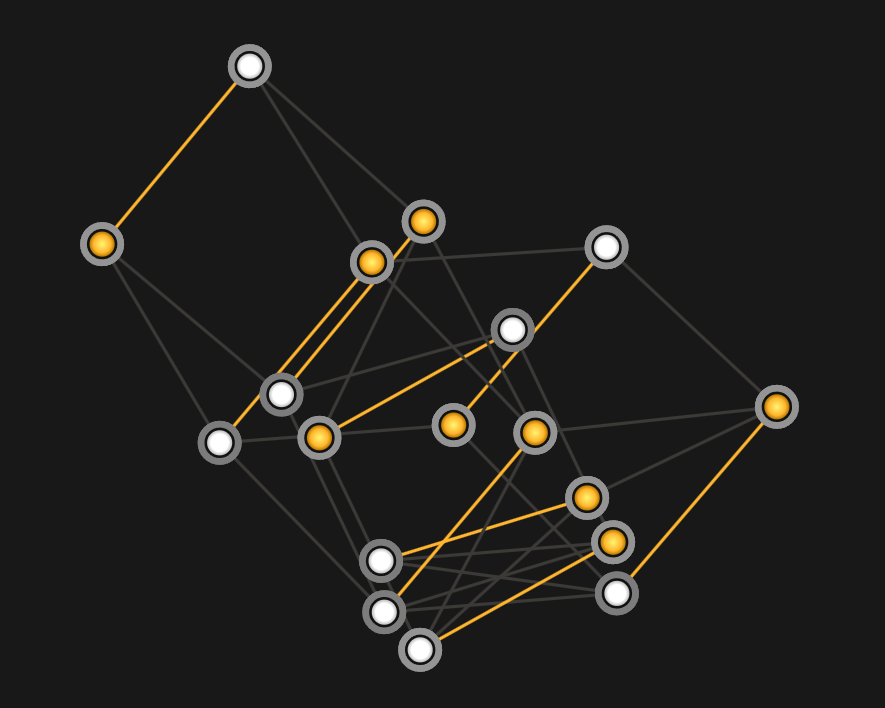}
                 \caption{}
                 \label{fig_graphs_a}
             \end{subfigure}
        \hspace*{0.2em}
        \begin{subfigure}{0.4\textwidth}
                 \centering
                 \includegraphics[width=\textwidth]{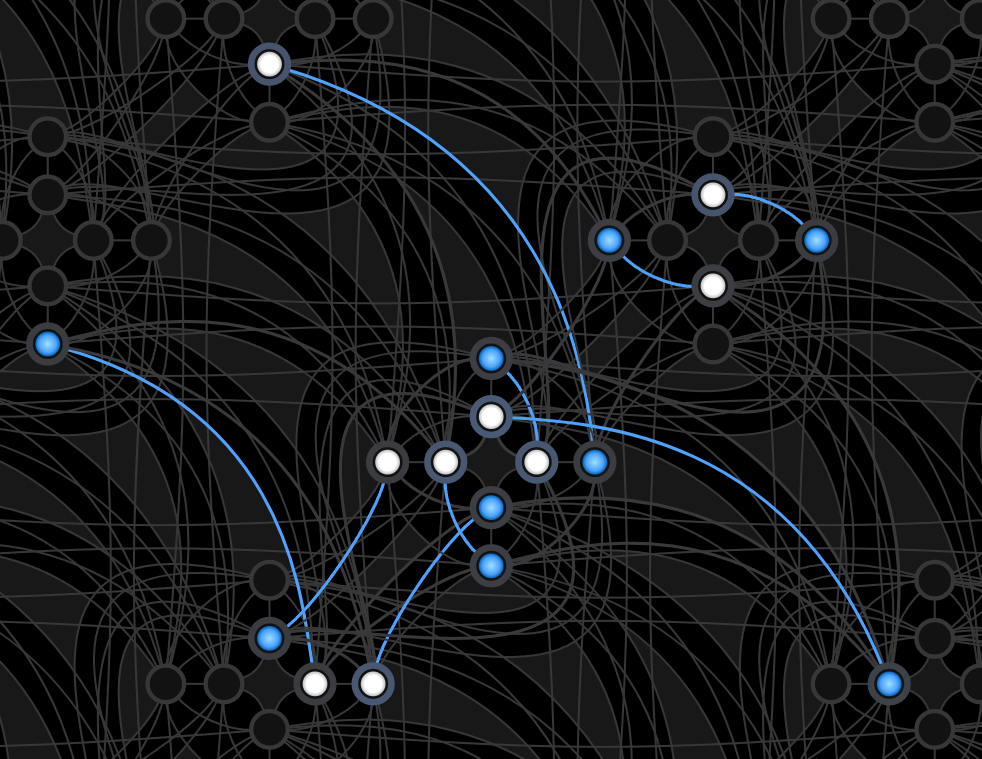}
                 \caption{}
                 \label{fig_graphs_b}
             \end{subfigure}
        \caption {D-Wave minor embedding for the given QUBO example. (a) the QUBO graph, and (b) the QPU graph.}
        \label{fig_graphs}
    \end{figure*}

\end{document}